\documentclass[1pt]{article}
\usepackage{hanging}
\usepackage{amsmath}
\usepackage{amsfonts,amssymb,amsthm,overpic}
\makeatletter
\newcommand{\diam}{\mathop{\operator@font diam}}
\usepackage{setspace}
\usepackage{multicol}
\usepackage{multirow}
\usepackage[compact]{titlesec}

\usepackage{epsfig}

\newtheorem{definition}{Definition}[section]
\newtheorem{remark}{Remark}[section]

\newtheorem{proposition}{Proposition}[section]

\newtheorem{lemma}{Lemma}[section]

\begin{document}

\title{\Huge{{On a duality between time and space cones }}}

\author{Waleed Al-Qallaf, Kyriakos Papadopoulos* \\
\small{Department of Mathematics, Kuwait University, Safat 13060, Kuwait}\\
{\it E-mail*: \textrm{ kyriakos@sci.kuniv.edu.kw}}}

\date{}

\maketitle

\begin{abstract}
We give an exact mathematical construction of a spacelike order $<$, which is dual to the
standard chronological order $\ll$ in the $n$-dimensional Minkowski space $M^n$,
and we discuss its order-theoretic, geometrical as well as its topological implications, conjecturing a possible extension to curved spacetimes.
\end{abstract}

\begin{multicols}{2}

\noindent{{\bf{1. Introduction}}}

\noindent{Many important theorems, within the frame of general relativity,
refer to spacelike properties, such as singularity theorems which require the
existence of a smooth spacelike Cauchy surface $\Sigma$ (Hawking et al., 1973)
, schematic conformal diagrams depicting
causal independence (for example, Penrose, 2007), etc. In all  cases,
spacelike is synonymous to locally acausal \footnote{It was highlighted to us by a reviewer of this article, and we consider it useful to mention this here as well, that spacelike is a purely local property, but acausal is a global property. For instance, the Lorentzian cylinder $M = \mathbb{R}^1 \times \mathbb{S}^1$, with metric $-dt^2 + d\theta^2$, and also the submanifold that is the image of the map $f: \mathbb{R} \to M$ given by $f(s) = (\frac{1}{2} s,s)$: this is manifestly spacelike at all points, but $f(s) \ll f(s+2\pi)$ for all $s$.}, where there is no timelike relation
or information traveling to the speed of light. In this article,
we show that the structure of the null-cone is induced, in a topological sense, by a spacelike order which
creates a spacelike orientation in an analogous way to the timelike
orientation. }

\noindent{{\bf{2. Definitions and notation}}}

\noindent{Let $M^n$ be the $n$-dimensional Minkowski space. Let $Q$ be
the characteristic quadratic form on $M^n$, defined by
$Q(x) = \{x_0^2 + x_1^2 + x_2^2 + \cdots +x_{n-2}^2 - x_{n-1}^2: x = (x_0,x_1,x_2,\cdots,x_{n-1}) \in M^n\}$.

For an event $x \in M^n$, we consider the following sets:
\begin{enumerate}

\item $C^T(x) = \{y : y=x \textrm{ or } Q(y-x) < 0\}$, the {\em time-cone} of $x$,

\item $C^L (x) = \{y : Q(y-x) = 0\}$, the {\em light-cone} of $x$,

\item $C^S(x) = \{y : y=x \textrm{ or } Q(y-x) > 0\}$, the {\em space-cone} \footnote{Here the word ``cone'' is used in a generalised sense, i.e. it is a cone on $I \times \mathbb{S}^{n-2}$. } of $x$,

\item $C^{LT}(x) = C^T(x) \cup C^L(x)$, the union of the time- and light-cones of $x$, also known as the {\em causal cone} of $x$, and

\item $C^{LS}(x) = C^S(x) \cup C^L(x)$, the union of the space- and light-cones of $x$.

\end{enumerate} }


\noindent{For any plane $P_m(x)$, where
$m \in M^n$,  $m \neq 0$ is the normal to  $P_m(x)$, $P_m(x) = \{y : g(m,y-x)=0\}$ and
where $g$ denotes the spacetime metric, we  consider the half-planes: }

\begin{enumerate}
\item $P_m^+(x) = \{ y : g(m,y-x)\geq 0 \textrm{ and } y\neq x\}$ and

\item $P_m^-(x) = \{y : g(m,y-x)\leq 0 \textrm{ and } y\neq x\}$
\end{enumerate}

For abbreviation, we will write $P_+(x) := P_m^+(x)$, $P_-(x) := P_m^-(x)$
and $P(x) := P_m(x)$.

We observe that $P_+(x) \cup P_-(x) = \{x\}^c$, where the superscript $c$ -here and throughout the text- denotes the complement of a set. So, for an event $y \in M$,
$y \in P_+(x) \cup P_-(x) = \{x\}^c$ if and only if $y \neq x$.

Moreover, we define the following subspaces:

\begin{enumerate}

\item $C_+^L(x) = P_+(x) \cap C^L(x)$;

\item $C_-^L(x) = P_-(x) \cap C^L(x)$;

\item $C_+^S(x) = P_+(x) \cap C^S(x)$;

\item $C_-^S(x) = P_-(x) \cap C^S(x)$;


\item $C_+^{LS}(x) = P_+(x) \cap C^{LS}(x)$;

\item $C_-^{LS}(x) = P_-(x) \cap C^{LS}(x)$.

\end{enumerate}

\noindent{It is standard (see (Penrose, 1972)) to consider two partial orders, the {\em chronological} order $\ll$ (which is irreflexive in $M^n$) and
the {\em causal order} $\prec$, which is reflexive, two orders that are defined not only in $M^n$ but in general in any spacetime,
as follows: }
\begin{enumerate}

\item $x \ll y$ iff $y \in C_+^T(x)$ and

\item $x \prec y$ iff $y \in C_+^T(x) \cup C_+^L(x)$

\end{enumerate}

\noindent{In addition, the reflexive relation {\em horismos} $\rightarrow$ is defined as $x \rightarrow y \textrm{ iff } x \prec y \textrm{ but not } x \ll y $.}
\\
\\






\noindent{{\bf{3. The weak interval topology}}}

\noindent{Consider the {\em weak interval topology},
which is constructed in an analogous way to the interval topology (Gierz et al., 1980),
which however does not apply only to lattices. In fact, when restricted to the $2$-dimensional
Minkowski space $M^2$, under the causal order $\prec$, the weak interval topology coincides with the interval topology,
but in general it will not be restricted to lattices. For its construction, we need
a relation $R$ defined on a set $X$. We then consider the sets $I^+(x) = \{y \in X : x R y\}$
and $I^-(x) = \{y \in X : y R x\}$, as well as the collections $\mathcal{S}^+ = \{X\setminus I^-(x) : x \in X\}$
and $\mathcal{S}^- = \{X\setminus I^+(x) : x \in X\}$. A basic-open set $U$ in the weak interval topology $T^{in}$
is defined as $U = A \cap B$, where $A \in \mathcal{S}^+$ and $B \in \mathcal{S}^-$; in other
words, $\mathcal{S}^+ \cup \mathcal{S}^-$ forms a subbase for $T^{in}$. }

\noindent{
The topology $T^{in}$ with respect to the relation $\rightarrow$ was studied
in Antoniadis et al., (2016) and in Papadopoulos et al., (2018a) 
and with respect to the order $\ll$, Papadopoulos et al. (2018b)
 conjectured there is a possibility of creating a duality between timelike
and spacelike, based on a spacelike order $<$ dual to chronology $\ll$. This
paper shows the exact mathematical construction of $<$ as well as the induced topology $T_<^{in}$,
for $n$-dimensional Minkowski space $M^n$. }

\noindent{{\bf{4. The order on the space-cone and its induced topology}}}

\noindent{We define a partial spacelike order $<$ dual to the chronological
order $\ll$. This order is obviously not causal, but it brings
an interesting duality between the time cone $C^T(x)$ and the
space cone $C^S(x)$ of an event $x$. Through $<$, the ``cone'' $C^S(x)$
exhibits similar properties to $C^T$ under $\ll$. Since ``chronological'' comes from the
word ``chronos'', which means time, we name $<$ ``chorological'', as it
refers to ``choros'', space.}

\begin{definition}
For non causally-related events $x,y \in M^n$, $x<y$ iff $y \in C_+^S(x)$, where we have defined $C_+^S(x)$ for some fixed choice of $m \in M^n$.
\end{definition}

\noindent{It follows that $x<y$ iff $x \in C_-^S(y)$. In addition, $\leq$  denotes $<$ including the boundary, in a dual way as $\prec$ is to $\ll$, that is, $x \leq y$ iff $y \in C_+^{LS}(x)$. }

\noindent{
We remark that $<$ is a partial order; the transitivity is obvious, as soon as it is highlighted that $<$ refers to events which are not causally related; thus, if $x,y,z$ are mutually not causally related ($x$ is not causally related to $y$, $y$ is not causally related to $z$ and $x$ is not causally related to $z$), then $x<y$ and $y<z$ implies that $x<z$.  }

\begin{lemma}\label{beforemain}
$C_+^{LS}(x) \cup C_-^{LS}(x) = C^{LS}(x) -\{x\}$.
\end{lemma}

\begin{proof}
$C_+^{LS}(x) \cup C_-^{LS}(x) = C^{LS}(x) \cap (P_+(x) \cup P_-(x)) = C^{LS}(x) \cap \{x\}^c = C^{LS}(x) - \{x\}$.

\end{proof}

\begin{proposition}\label{Main}
$[C_+^{LS}(x)]^c \cap [C_-^{LS}(x)]^c = C^T(x)$
\end{proposition}

\begin{proof}
$[C_+^{LS}(x)]^c \cap [C_-^{LS}(x)]^c$\\
$=[C_+^{LS}(x) \cup C_-^{LS}(x)]^c$\\
$=[C^{LS}(x) - \{x\}]^c$\\
$=[C^{LS}(x)]^c \cup \{x\}$\\
$=[C^T(x) - \{x\}] \cup \{x\} $\\
$=C^T(x)$
\end{proof}

\begin{remark}
The vector $m$ (normal to plane $P$) is perpendicular to the vector $<x_0,x_1,...,x_{n-1}>=<0,..,0,1>$ where $x_{n-1}$ corresponds to time.

Also, both Lemma \ref{beforemain} and Proposition \ref{Main}, if expressed for $C^S$ instead of for $C^{LS}$, i.e. if the boundary was omitted, would give $[C_+^S(x)]^c \cap [C_-^S(x)]^c$ which would equal $C^L(x) \cup C^T(x)$, i.e the causal cone $C^{LT}(x)$. 
\end{remark}

\begin{proposition}\label{Main2}
The weak interval topology $T_\leq^{in}$,
generated by the space-like order $\leq$, has as basic-open sets the time-cones
$C^{T}(x)$.
\end{proposition}
\begin{proof}
The proof follows from Proposition \ref{Main},
by observing that $[C_+^{LS}(x)]^c \in \mathcal{S}^-$ and $[C_-^{LS}(x)]^c \in \mathcal{S}^+$,
where $I^+(x)$ and $I^-(x)$ are defined with respect to $R = \leq$ (see Section 3).
\end{proof}

\noindent{For the last result of our discussion, we will need to use Reed's definition
of {\em intersection topology} (Reed, 1986): }

\begin{definition}\label{intersection topology}
If $T_1$ and $T_2$ are two topologies on a set $X$, then the {\em intersection topology} $T^{int}$
with respect to $T_1$ and $T_2$, is the topology on $X$ such that the set $\{U_1 \cap U_2 : U_1 \in T_1, U_2 \in T_2\}$
forms a base for $(X,T)$.
\end{definition}

\noindent{The topology $Z^T$ is defined to be the intersection topology, according to Reed's definition,
of the topologies $T_\leq^{in}$ and the natural topology of $\mathbb{R}^n$, in $M^n$. This topology, in $M$, coincides
with one of the three topologies that were suggested by Zeeman (Zeeman, 1967), as alternatives
to his Fine topology. Its characteristic is that its open sets are time-cones bounded by
Euclidean-open balls, in $M^n$, and its general relativistic analog is actually the
Path topology of Hawking-King-McCarthy  (Hawking et al., 1976). Agrawal and Shrivastava reviewed several topological properties of $Z^T$ (Agrawal et al., 2009) 
showing that, due to its equality to the Path topology, it is Hausdorff, separable, first countable, and path-connected, not regular, not metrisable, non-Lindel\"of and not simply connected. In addition, the authors completed an in depth study of Zeno sequences in $Z^T$. Zeeman introduced Zeno sequences (Zeeman, 1967) with respect to his ``Fine'' topology $F$; a sequence $\{z_n\}_{n \in \mathbb{N}}$ which converges to some $z$ in $M^n$ under the natural topology of $\mathbb{R}^n$ and not under the topology $F$ (or any other topology in the class of Zeeman topologies, (G\"obel, 1976)) is called a Zeno sequence.  Agrawal-Shrivastava also showed that, within the $n$-dimensional Minkowski space $M^n$, for a Zeno sequence under topology $Z^T$ converging to $z \in M^n$ there exists a subsequence of this sequence whose image is closed under $Z^T$ but not under the natural topology of $\mathbb{R}^n$. In addition, within $M^n$ and for a nonempty open-set $G$ in the natural topology of $\mathbb{R}^n$, if $z\in G$, then $G$ contains a completed image of a Zeno sequence under $Z^T$ converging to $z$. With respect to the convergence of causal-curves, Low (2016) has shown that under the Path topology, i.e. under the general-relativistic analogue of $Z^T$, the Limit Curve Theorem fails to hold (Low, 2016). Thus, the basic arguments for building contradiction in singularity theorems fail under the Path topology, as well. }

\noindent{Low (2016) also establishes, with respect to the Path topology, that if one considers timelike paths, the notion of convergence is not affected by the choice of whether one uses the manifold topology or the path topology on a spacetime (Propositions 1 and 2). In the case of $M^n$, one can restate this argument by substituting the Path topology with $Z^T$ and the manifold topology with the natural topology of $\mathbb{R}^n$. In particular, let $\mathcal{T}$ be the set of (endless) timelike curves; these are curves in $M^n$ where at each point the tangent vector is future pointing and timelike, (Penrose, 1972). Also, $\mathcal{C}$ denotes the set of (endless) causal curves; these are curves in $M^n$ where at each point the tangent vector is future pointing and timelike or null. Then the restriction of $Z^T$ in $\mathcal{T}$ gives a topological space homeomorphic to the restriction of $\mathcal{T}$ to the natural topology of $\mathbb{R}^n$, and also if $\gamma_n \to 
\gamma$, some $\gamma \in \mathcal{T}$, is the restriction of $\mathcal{T}$ to the natural topology of $\mathbb{R}^n$, then for each $x \in \gamma$ there exists a sequence $x_n \in 
\gamma_n$, such that $x_n \to x$ in $Z^T$.}

\noindent{The topology $Z^{LT}$ is defined to be the intersection topology, according to Reed's definition,
of the topologies $T_<^{in}$ and the natural topology of $\mathbb{R}^n$, in $M^n$. This topology fully incorporates the
causal structure of $M^n$. This is so, because it admits a base of open sets of the form $C^{LT}(x) \cap B_{\epsilon}(x)$, where  $B_{\epsilon}(x)$ is a ball in $\mathbb{R}^n$ centered at $x$ and of radius $\epsilon >0$. Low, again with respect to curved spacetimes, shows that the Path topology induces a strictly finer topology on $\mathcal{C}$ than the manifold topology does and that the restriction of the manifold topology in $\mathcal{T}$ is dense in the restriction of the manifold topology in $\mathcal{C}$. Since there is no reference on a general relativistic analog of $Z^{LT}$, we can make the following considerations with respect to the statements of part IV of (Low, 2016), before Proposition 3. Consider $x \in \gamma$ and let a neighborhood of $x$ in $\gamma$ with respect to the natural topology of $\mathbb{R}^n$ be a null-geodesic segment. Then there exists no sequence of timelike curves such that $\gamma_n \in \gamma$ with respect to the restriction of $\mathcal{C}$ in the natural topology but not in $Z^{LT}$. So, the general relativistic analogue of $Z^{LT}$ will not induce a strictly finer topology on $\mathcal{C}$ than the manifold topology does.}

\noindent{{\bf{5. Questions}}}

\noindent{It would be desirable if the  results of Section 4 generalized to any curved spacetime, in the frame of general relativity. This hope comes for the following
intuition. In a relativistic spacetime manifold, wherever there is spacetime, there are events and for every event
there is a light-cone. Since
our construction of $T_\leq^{in}$ is topological, depending exclusively on the interior (time-cone), boundary (light-cone) and
exterior (space-cone) of an event $x$, independently of the geometry of the space-time, one could consider the general-relativistic analog of $T_\leq^{in}$ as the topology which has as basic-open sets the time-cones $C^T(x)$, and the general-relativistic analogue of $Z^T$ as the topology which has basic-open sets the bounded time-cones
$C^T(x) \cap B_{\epsilon}^h(x)$, for some Riemannian metric $h$ on the spacetime manifold, and by considering Riemann-open balls $B_{\epsilon}^h(x)$. This is trivial, from a topological perspective, when we already know the open sets of the special-relativistic topology $Z^T$ without needing any information about the order $\leq$.}

\noindent{The question in this case is how could one define the general-relativistic analogue of the order $\leq$.
How could one describe the general-relativistic analogues of the half-planes $P_+(x)$ and $P_-(x)$, that we examined in Section 2, since they will not be ``flat'' planes in the Euclidean sense anymore, but will follow the geometry of the particular spacetime manifold, so that their union will give ${\{x}\}^c$. So, it would be vital to also express the general relativistic analogue of the normal $m$ to the plane $P(x)$ (Section 2), in a rigorous algebraic way; such an algebraic development should open further directions to our discussion about the duality between causal and locally acausal orders in a spacetime, a duality which might play a role for the passage from locality to non-locality (see, for example, (Vagenas, 2018)). An answer to such a question will also give a solution to the orderability problem (see (Papadopoulos, 2014)) in the particular case of the Path topology of Hawking-King-McCarthy. Similar questions may be asked for the general relativistic analogue of $Z^{LT}$ and a possible generalisation of the order $<$ to curved spacetimes. }

\noindent{{ ACKNOWLEDGEMENTS} }

\noindent{Even if this work reserves a note and is just an initial step on a certainly big discussion, the authors are grateful to many people for shaping their ideas; we are grateful to Fabio Scardigli, for the introduction to the (endless) bibliography and to Nikolaos Kalogeropoulos who gave some hints on how a study of spacelike-timelike duality could play a role in the quantum theory of gravity. Many thanks also to Robert Low for the useful comments on an earlier draft.}

\end{multicols}

\section*{References}

\begin{hangparas}{.25in}{1}

{\bf{Agrawal, G. \& Shrivastava, S. (2009).}}
$t$-topology on the $n$-dimensional Minkowski space. Journal of Mathematical Physics 50, 053515.

{\bf{Antoniadis, I. \& Cotsakis, S. \& Papadopoulos, K. (2016).}} The Causal Order on the Ambient Boundary. Mod. Phys. Lett. A, Vol 31, Issue 20.

{\bf{Gierz, G. \& Hofmann, K.H. \& Keimel, K. \& Lawson, J.D. \& Mislove, M.W. \& Scott, D.S. (1980).}} A compendium of continuous lattices. Springer-Verlag.

{\bf{G\"obel, R. (1976).}} Zeeman Topologies on Space-Times of General Relativity Theory. Comm. Math. Phys. 46, 289-307.

{\bf{Hawking, S.W. \& Ellis, G.F.R. (1973).}} The Large Scale Structure of Space-Time. Cambridge University Press.

{\bf{Hawking, S.W. \& King, A.R. \& McCarthy, P.J. (1976).}} A new topology for curved space–time which incorporates the causal, differential, and conformal structures. Journal of Mathematical Physics, 17 (2). pp. 174-181.

{\bf{Low, R.J. (2016).}} Spaces of Paths and the Path Topology. Journal of Mathematical Physics 57, 092503.

{\bf{Papadopoulos, K. (2014).}} On the Orderability Problem and the Interval Topology. Chapter in the Volume ``Topics in Mathematical Analysis and Applications'', in the Optimization and Its Applications Springer Series, T. Rassias and L. Toth Eds, Springer Verlag.

{\bf{Papadopoulos, K. \& Acharjee, S. \& Papadopoulos, B.K. (2018a).}} The Order On the Light Cone and Its Induced Topology. International Journal of Geometric Methods in Modern Physics 15, 1850069.

{\bf{Papadopoulos, K. \& Papadopoulos, B.K. (2018b).}}
On Two Topologies that were suggested by Zeeman. Mathematical Methods in the Applied Sciences, Vol. 41, Issue 17.

{\bf{Penrose, R. (1972).}} Techniques of Differential Topology in Relativity. CBMS-NSF Regional Conference Series in Applied Mathematics.

{\bf{Penrose, R. (2007).}} The Road to Reality: a complete guide to the laws of the universe. Vintage Books.

{\bf{Reed, G.M. (1986)}} The intersection topology w.r.t. the real line and the countable ordinals. Trans. Am. Math. Society, Vol. 297, No 2, pp 509-520.

{\bf{Vagenas, E., (2018).}} Can an Axion be the Dark Energy Particle?. Kuwait J.Sci., 45, 53 -56.

{\bf{Zeeman, E.C., (1967).}} The Topology of Minkowski Space. Topology, Vol. 6, 161-170.
\end{hangparas}
\end{document}